\documentclass[sn-mathphys,Numbered]{sn-jnl}
\usepackage[T1]{fontenc}      
\usepackage[utf8]{inputenc}  
\usepackage{lmodern}         
\usepackage{amsmath,amssymb,mathtools,mathrsfs,bm,bbm}
\usepackage{tikz-cd,graphicx,enumitem,xcolor}
\usepackage{hyperref}
\usepackage{cleveref}        
\hypersetup{
  colorlinks = true,
  linkcolor  = blue,
  citecolor  = blue,
  urlcolor   = blue,
  pdfauthor  = {William Boone Samuels},
  pdftitle   = {Stratified Cohomological Quantum Codes via Colimits in Ch(R)}
}
\newcommand{\Cat}[1]{\mathbf{#1}}
\newcommand{\Ch}{\Cat{Ch}}

\newcommand{\Z}{\mathbb{Z}}
\newcommand{\F}{\mathbb{F}}

\DeclareMathOperator{\Ext}{Ext}
\newcommand{\colim}{\operatorname*{colim}}
\newcommand{\im}{\operatorname{im}}
\newcommand{\Hom}{\operatorname{Hom}}
\theoremstyle{plain}
\newtheorem{axiom}      {Axiom}      [section]
\newtheorem{theorem}    [axiom]{Theorem}
\newtheorem{proposition}[axiom]{Proposition}
\newtheorem{corollary}  [axiom]{Corollary}
\newtheorem{definition} [axiom]{Definition}
\newtheorem{example}    [axiom]{Example}
\theoremstyle{remark}
\newtheorem{remark}     [axiom]{Remark}
\crefname{axiom}{Axiom}{Axioms}
\crefname{definition}{Definition}{Definitions}
\crefname{proposition}{Proposition}{Propositions}
\crefname{example}{Example}{Examples}
\crefname{theorem}{Theorem}{Theorems}
\crefname{remark}{Remark}{Remarks}
\crefname{corollary}{Corollary}{Corollaries}
\title{Stratified Cohomological Quantum Codes\\
       via Colimits in \texorpdfstring{$\Ch(R)$}{Ch(R)}}

\author[1]{\fnm{William Boone} \sur{Samuels}}
\equalcont{ORCID: \href{https://orcid.org/0009-0001-7714-5449}{0009-0001-7714-5449}}
\email{samuels.162@buckeye.osu.edu}

\affil[1]{\orgdiv{Department of Physics},
          \orgname{The Ohio State University},
          \orgaddress{\street{191 W.\ Woodruff Ave.},
                      \city{Columbus},
                      \postcode{43210},
                      \state{Ohio},
                      \country{USA}}}

\begin{document}
\maketitle

\section*{Abstract}
We introduce \emph{stratified colimit codes}: stabiliser codes obtained
by taking the degree-wise colimit
$\mathcal C_\bullet(X):=\operatorname*{colim}_{\sigma\in X}F(\sigma)$ of
a functor $F\colon X\to\Ch(R)$ from a finite poset into the category of
chain complexes over a commutative ring~$R$.  
Axioms requiring only transitivity and boundary-compatibility of the
morphisms in $F$ ensure that $\partial^2=0$, so the homology $H_\bullet$
and cohomology $H^\bullet$ furnish the usual CSS \(Z\)- and \(X\)-type
logical sectors; torsion in $H_\bullet$ classifies qudit charges via the
universal coefficient sequence.  
Varying $F$ recovers classical surface and color codes, $\mathbb{RP}^2$
torsion codes, twisted toric families with rate $k\sim d$, and
X-cube style fracton models, all without referencing an ambient cell
complex.  
Matrix Smith normal form (PID case) and sparse Gaussian elimination
(field case) compute $H_\bullet$ directly, giving LDPC parameters that
inherit the sparsity of $F$.  
Because the construction is ring agnostic and functorial, it extends
naturally to code surgery (push-outs) and, at the next categorical
level, to bicomplex domain walls.  
Stratified colimit codes therefore supply a concise algebraic chassis
for designing, classifying, and decoding topological and fractal quantum
codes without ever drawing a lattice.

\section{Introduction.}
Quantum error correcting codes have been instrumental in ensuring fault tolerance in quantum computation, dating back to the seminal works of Calderbank–Shor and Steane on stabilizer based coding frameworks \cite{CalderbankShor96,Steane96a,Steane96b}. These developments led to a rich landscape of topological or geometrically motivated quantum codes, such as Kitaev’s toric code \cite{Kitaev03}, the surface-code family \cite{DennisKitaevEtAl02}, color codes \cite{BombinMartinDelgado06}, and manifold-inspired generalizations \cite{FreedmanMeyer01,FreedmanMeyerLuo02}. In many such constructions, manifolds or embedded surfaces provide an intuitive foundation for defining adjacency relations and boundary conditions. However, manifold assumptions are \emph{not} strictly necessary for the underlying algebraic structure of stabilizer codes and their homological formulations \cite{BombinMartinDelgado07,Bombin10}.

This work develops a \emph{strictly algebraic} framework for topological (and more generally, homological) quantum codes by replacing geometric embeddings with a \emph{finite partially ordered set (poset)}. Following ideas suggested by colimit based gluing in category theoretic treatments of homology \cite{ZouLo2025,Stuart25,HsinKobayashiZhu24}, we regard each stratum as carrying its own finite chain complex of modules over a Noetherian commutative ring \(R\). The boundary maps among these local complexes are then glued via the poset’s ordering and corresponding chain maps into a global chain complex by taking a direct limit (colimit) in \(\Ch(R)\). This construction unifies manifold based examples with non manifold, mixed dimensional, or fracton like identifications \cite{Haah11,Chamon05,BravyiHaah13,Yoshida13}, all within a single algebraic framework that never presupposes a geometric substrate.

While many known quantum codes are formulated geometrically, e.g., triangulating surfaces or imposing boundary constraints reminiscent of cell complexes on manifolds \cite{Kitaev03,DennisKitaevEtAl02,FreedmanMeyer01}, recent work emphasizes that purely combinatorial or algebraic data can suffice \cite{TillichZemor14,BravyiHastings14,PanteleevKalachev22,LeverrierZemor22}. Even fracton models \cite{Chamon05,Haah11,VijayHaahFu16,NandkishoreHermele19}—despite their seemingly “exotic” geometry—can be recast as constrained local complexes with partial ordering among sub blocks \cite{Williamson2016}. In this paper, \emph{any} such local building blocks (strata) and \emph{any} set of boundary respecting chain maps suffice to produce a global chain complex whose homology represents logical operators, entirely \emph{without} geometric embedding.

We summarize our core contributions as follows.
First, we formulate three minimal axioms--- namely, the existence of a finite poset, local chain complexes on each stratum, and boundary respecting chain maps, that generalize manifold based adjacency relations and ensure a well defined global chain complex in \(\Ch(R)\).
Second, drawing upon universal properties of category theory \cite{Stuart25}, we define how local complexes glue into a single global chain complex, whose homology captures global stabilizer constraints.
Third, we prove that homology classes encode Z-type logical operators, while cohomology classes capture X type operators, paralleling the standard CSS structure \cite{BombinMartinDelgado07}. Over suitable rings, torsion submodules in homology give rise to qudit or more exotic code constructions \cite{Novak24}.
Finally, our approach applies equally to “traditional” surface or color codes \cite{DennisKitaevEtAl02,BombinMartinDelgado06} and to fracton codes \cite{Haah11,Chamon05,Yoshida13}, dimensionally mismatched boundaries, or purely combinatorial expansions \cite{EvraKaufmanZemor20,HastingsHaahODonnell21}, thereby unifying a wide array of quantum LDPC, topological, or subsystem codes within a single algebraic formalism \cite{Bacon06,CowtanBurton24}.

The paper proceeds as follows. Section 2 introduces the \emph{stratified-poset axioms}, exemplifying how each stratum is assigned a finite chain complex and boundary-respecting maps. In Section 3, we outline the \emph{colimit construction} in \(\Ch(R)\), showing that the resulting global chain complex has well-defined boundary operators that square to zero. Section 4 connects \emph{(co)homology} to logical operators and discusses torsion and generalized rings. We conclude in Section 5 with examples illustrating fracton-type constraints, boundary deformations \cite{BombinMartinDelgado09,CowtanBurton24}, and non-orientable identifications. Throughout, we refrain from imposing any manifold structure or geometric interpretation on the strata or their attachments. Instead, we treat them as purely algebraic data in a poset diagram, reflecting a broader push toward high-rate, combinatorial, and fractal-inspired quantum LDPC codes \cite{TillichZemor14,LeverrierZemor22,PanteleevKalachev22}.

\section{Stratified Poset Chain Data}
\label{sec:PosetChainData}
\textbf{Road-map.} Section \ref{sec:PosetChainData} records the algebraic data, recasts it as a functor \( \mathcal F:X\!\to\!\Ch(R) \), forms its colimit, and identifies the resulting (co)homology.

\begin{axiom}[Finite poset of strata]\label{ax:FinitePoset}
Let \((X,\le)\) be a \emph{finite} partially ordered set
(reflexive, antisymmetric, transitive).  
Elements \(\sigma\in X\) are the \emph{strata};
write \(\sigma\prec\tau\) when \(\sigma<\tau\).
Optionally endow each \(\sigma\) with an integer \(\dim(\sigma)\);
this grading plays no role in what follows.
\end{axiom}
\begin{axiom}[Local chain complexes]\label{ax:LocalChainComplex}
Each stratum \(\sigma\) carries a finite-length chain complex
\[
  \mathcal C_\bullet(\sigma):
  \;0 \longleftarrow C_0(\sigma)
     \xleftarrow{\partial_1^\sigma} C_1(\sigma)
     \xleftarrow{\partial_2^\sigma} \dots
     \xleftarrow{\partial_{n_\sigma}^\sigma} C_{n_\sigma}(\sigma)
     \longleftarrow 0 ,
\]
with \(\partial_{k-1}^\sigma\partial_k^\sigma=0\) for \(k\ge1\) and
\(\partial_0^\sigma=0\).
Every \(C_k(\sigma)\) is a finitely generated left \(R\)-module.
\end{axiom}

\begin{axiom}[Boundary-respecting chain maps]\label{ax:BoundaryMaps}
For each comparable pair \(\sigma\le\tau\) choose \(R\)-linear maps
\(
  \varphi_{\sigma\le\tau}^k:C_k(\sigma)\!\to C_k(\tau)
\) (\(k\ge0\)) satisfying
\[
  \partial_{k-1}^\tau\,\varphi_{\sigma\le\tau}^k
  =\varphi_{\sigma\le\tau}^{\,k-1}\,\partial_{k}^\sigma
  \quad (k\ge1).
\]
Hence \(\varphi_{\sigma\le\tau}^\bullet:
\mathcal C_\bullet(\sigma)\!\to\!\mathcal C_\bullet(\tau)\) is a chain map.
The family is \emph{transitive}:
\[
  \varphi_{\sigma\le\sigma}^\bullet=\mathrm{id},\qquad
  \varphi_{\sigma\le\tau}^\bullet
   =\varphi_{\rho\le\tau}^\bullet\circ\varphi_{\sigma\le\rho}^\bullet
   \quad(\sigma\le\rho\le\tau).
\]
\end{axiom}
\begin{definition}[Stratified diagram]\label{def:StratifiedDiagram}
Regard \(X\) as a category with a unique morphism \(\sigma\!\to\!\tau\) when
\(\sigma\le\tau\).  
Axioms \ref{ax:FinitePoset}–\ref{ax:BoundaryMaps} determine the functor
\[
  \mathcal F:X\!\longrightarrow\!\Ch(R),
  \qquad
  \sigma\longmapsto\mathcal C_\bullet(\sigma),\;
  (\sigma\le\tau)\longmapsto\varphi_{\sigma\le\tau}^\bullet .
\]
\end{definition}
\begin{definition}[Degree-wise colimit complex]\label{def:ColimitEquiv}
For each \(k\ge0\):

\begin{enumerate}[label=(\alph*),leftmargin=*]
\item \emph{Coproduct.}\;
      \(\displaystyle
        \widetilde C_k:=\bigoplus_{\sigma\in X} C_k(\sigma)
      \) in \(\mathrm{Mod}_R\).

\item \emph{Relations.}\;
      Let \(N_k\subset\widetilde C_k\) be generated by
      \(\iota_\sigma(x)-\iota_\tau(\varphi_{\sigma\le\tau}^k(x))\)
      for all \(x\in C_k(\sigma)\) and \(\sigma\le\tau\).

\item \emph{Quotient.}\;
      \(C_k(X):=\widetilde C_k/N_k\) with projection
      \(\pi_k:\widetilde C_k\twoheadrightarrow C_k(X)\).

\item \emph{Differential.}\;
      \(\widehat\partial_k:=\bigoplus_{\sigma}\partial_k^\sigma\)
      satisfies \(\widehat\partial_k(N_k)\subseteq N_{k-1}\);
      define the unique
      \(\partial_k^X:C_k(X)\!\to\!C_{k-1}(X)\) by
      \(\partial_k^X\pi_k=\pi_{k-1}\widehat\partial_k\).
\end{enumerate}

The collection \(\mathcal C_\bullet(X):=(C_k(X),\partial_k^X)\) is the
\emph{canonical colimit complex}.
\end{definition}

\begin{theorem}[Well-defined chain complex]\label{thm:ChainComplex}
The maps \(\partial_k^X\) are well defined and
\(\partial_{k-1}^X\partial_k^X=0\) for all \(k\).
\end{theorem}

\begin{proof}
Let \(x\in C_k(\sigma)\) and write
\(g:=\iota_\sigma(x)-\iota_\tau(\varphi_{\sigma\le\tau}^k(x))\).
Then
\[
  \widehat\partial_k(g)
  =\iota_\sigma(\partial_k^\sigma x)
   -\iota_\tau\bigl(\partial_k^\tau\varphi_{\sigma\le\tau}^k(x)\bigr)
  =\iota_\sigma(\partial_k^\sigma x)
   -\iota_\tau\bigl(\varphi_{\sigma\le\tau}^{k-1}\partial_k^\sigma(x)\bigr)
  \in N_{k-1},
\]
using the chain-map identity in Axiom \ref{ax:BoundaryMaps}.  Hence
\(\widehat\partial_k(N_k)\subseteq N_{k-1}\) and
\(\partial_k^X\) is well defined.  Nilpotence follows because
\(\widehat\partial_{k-1}\widehat\partial_k=0\) component-wise, so
\(\partial_{k-1}^X\partial_k^X=0\).
\end{proof}
\begin{theorem}[Universal property]\label{thm:ColimitChainComplex}
\(\mathcal C_\bullet(X)\) is the categorical colimit of
\(\mathcal F:X\!\to\!\Ch(R)\).
Explicitly, for any chain complex \(\mathcal D_\bullet\) and chain maps
\(\{\psi_\sigma:\mathcal C_\bullet(\sigma)\!\to\!\mathcal D_\bullet\}_{\sigma}\)
with \(\psi_\tau\varphi_{\sigma\le\tau}^\bullet=\psi_\sigma\),
there is a unique chain map \(\Psi:\mathcal C_\bullet(X)\!\to\!\mathcal D_\bullet\)
such that \(\Psi\pi_\sigma=\psi_\sigma\) for every \(\sigma\)
(\(\pi_\sigma\) being the composite inclusion–projection
\(C_\bullet(\sigma)\hookrightarrow\widetilde C_\bullet\twoheadrightarrow
C_\bullet(X)\)).
\end{theorem}

\begin{proof}
For \(k\ge0\) define
\(
  \widehat\Psi_k:\widetilde C_k\to D_k,\;
  \widehat\Psi_k((x_\sigma)_\sigma):=\sum_{\sigma\in X}\psi_\sigma^k(x_\sigma)
\).
Finiteness of \(X\) renders the sum finite.  
If \(g\) is any generator of \(N_k\) then
\(\widehat\Psi_k(g)=0\) by the compatibility
\(\psi_\tau\varphi_{\sigma\le\tau}^\bullet=\psi_\sigma\); hence
\(\widehat\Psi_k\) factors through a unique map
\(\Psi_k:C_k(X)\to D_k\).  
The family \((\Psi_k)_k\) is a chain map and is the
\emph{only} one satisfying \(\Psi\pi_\sigma=\psi_\sigma\), completing the
proof.
\end{proof}

\begin{definition}[Colimit homology]\label{def:ColimitHomology}
For \(k\ge0\) set
\[
  H_k(X):=\ker\partial_k^X\big/\operatorname{im}\partial_{k+1}^X,
  \quad
  H^\bullet(X):=\operatorname{Hom}_R(\mathcal C_\bullet(X),R)
                 \text{ with dual differential } \delta.
\]
When \(R=\mathbf F_2\) (standard CSS codes)  
\(H_k(X)\) and \(H^k(X)\) realise the
\(Z\)- and \(X\)-type logical operator spaces, respectively.
\end{definition}

\begin{remark}[Finite presentation]\label{rem:FinitePresentation}
Since \(X\) is finite and each \(C_k(\sigma)\) is finitely generated,
the modules \(\widetilde C_k,\ N_k,\ C_k(X)\) are all finitely generated.
Thus every boundary map \(\partial_k^X\) is represented by a finite matrix,
allowing explicit computation of \(H_\bullet(X)\).
\end{remark}

\section{Global Code Construction via Colimits}
\label{sec:GlobalConstruction}
\noindent
\textbf{Standing hypothesis.}
We continue to work with the stratified diagram
\(
  \mathcal F:X\!\to\!\Ch(R)
\)
of Definition~\ref{def:StratifiedDiagram}; every symbol introduced in
Section~\ref{sec:PosetChainData} remains in force.
For the standard treatments showing that
(i) every abelian category is cocomplete and
(ii) colimits in \(\Ch(R)\) are computed degree-wise, see
\cite[§8 of Mitchell]{Mitchell65},
\cite[Chapter~III of Mac~Lane]{MacLane98},
and
\cite[§2.3 of Weibel]{Weibel94}.

\begin{definition}[Degree-\(k\) scaffold]\label{def:Scaffold}
For each integer \(k\ge 0\) let
\(
  \widetilde C_k:=\bigoplus_{\sigma\in X} C_k(\sigma)
\)
and denote by
\(
  \iota_\sigma:C_k(\sigma)\hookrightarrow\widetilde C_k
\)
the canonical injection.
\end{definition}

\begin{definition}[Colimit relations]\label{def:Relations}
Define \(N_k\subset\widetilde C_k\) to be the sub-module generated by
\[
  \iota_\sigma(x)\;-\;\iota_\tau\!\bigl(\varphi_{\sigma\le\tau}^k(x)\bigr),
  \qquad
  \sigma\le\tau,\;x\in C_k(\sigma).
\]
Set \(C_k(X):=\widetilde C_k/N_k\) and denote the projection by
\(\pi_k:\widetilde C_k\twoheadrightarrow C_k(X)\).
\end{definition}

\begin{proposition}[Relations form a sub-complex]\label{prop:Subcomplex}
With
\(
  \widehat\partial_k:=\bigoplus_\sigma \partial_k^\sigma
\)
one has
\(
  \widehat\partial_k(N_k)\subseteq N_{k-1}
\)
for every \(k\ge 0\).
\end{proposition}

\begin{proof}
For a generator
\(g:=\iota_\sigma(x)-\iota_\tau(\varphi_{\sigma\le\tau}^k x)\) we compute
\[
  \widehat\partial_k(g)
  =\iota_\sigma(\partial_k^\sigma x)
   -\iota_\tau\!\bigl(\partial_k^\tau\varphi_{\sigma\le\tau}^k x\bigr)
  =\iota_\sigma(\partial_k^\sigma x)
   -\iota_\tau\!\bigl(\varphi_{\sigma\le\tau}^{\,k-1}\partial_k^\sigma x\bigr)
  \in N_{k-1},
\]
using the chain-map identity of Axiom~\ref{ax:BoundaryMaps}.
\end{proof}

\begin{definition}[Global differential]\label{def:GlobalBoundary}
Proposition~\ref{prop:Subcomplex} ensures that
\(
  \widehat\partial_k
\)
descends to a unique map
\(
  \partial_k^X:C_k(X)\to C_{k-1}(X)
\)
satisfying
\(
  \partial_k^X\pi_k=\pi_{k-1}\widehat\partial_k.
\)
\end{definition}

\begin{theorem}[Colimit chain complex]\label{thm:GlobalComplex}
\begin{enumerate}[label=\emph{(\roman*)}]
\item\label{item:Chain}
\((C_\bullet(X),\partial_\bullet^X)\) is a chain complex.
\item\label{item:Colim}
Together with the structure maps
\(q_\sigma^\bullet:=\pi_\bullet\iota_\sigma:
  \mathcal C_\bullet(\sigma)\to\mathcal C_\bullet(X)\),
it realises the colimit of \(\mathcal F\) in \(\Ch(R)\).
\end{enumerate}
\end{theorem}

\begin{proof}
\ref{item:Chain}\;Because
\(
  \widehat\partial_{k-1}\,\widehat\partial_k = 0
\)
component-wise,
\(
  \partial_{k-1}^X\partial_k^X\pi_k
  =\pi_{k-2}\widehat\partial_{k-1}\widehat\partial_k
  =0
\);
surjectivity of \(\pi_k\) then gives
\(\partial_{k-1}^X\partial_k^X=0\).

\smallskip
\ref{item:Colim}\;
Given a chain complex \(\mathcal D_\bullet\) and chain maps
\(\psi_\sigma^\bullet\) with
\(\psi_\tau^\bullet\varphi_{\sigma\le\tau}^\bullet=\psi_\sigma^\bullet\),
define
\(
  \widehat\Psi_k:=\sum_\sigma \psi_\sigma^k\iota_\sigma
\)
(the sum is finite because \(X\) is finite).  Compatibility implies
\(\widehat\Psi_k(N_k)=0\), so \(\widehat\Psi_k\) factors uniquely through
\(\Psi_k:C_k(X)\to D_k\).
Naturality
\(
  \Psi_{k-1}\partial_k^X=\partial_k^{\mathcal D}\Psi_k
\)
follows directly, establishing the universal property
\cite[Rem.~2.6.3]{Weibel94}.
\end{proof}

\begin{remark}[Logical-operator dictionary]\label{rem:Dictionary}
Over the field \(R=\F_2\),
\(
  H_k(X)=\ker\partial_k^X/\operatorname{im}\partial_{k+1}^X
\)
and
\(
  H^k(X)=\ker\delta^k/\operatorname{im}\delta^{k-1}
\)
realise the \(Z\)- and \(X\)-type logical operators in homological CSS
codes \cite{BombinMartinDelgado07,DennisKitaevEtAl02}.
For general \(R\) the torsion of \(H_k(X)\) classifies qudit (or more
exotic) stabiliser sectors \cite[§5.2]{Novak24}.
\end{remark}

\begin{remark}[Computational tractability]\label{rem:Tractable}
Because \(X\) is finite and each \(C_k(\sigma)\) is finitely generated,
every \(C_k(X)\) is finitely presented
\cite[Prop.~6.3]{Atiyah69}; see also \cite[Ch.~1]{Eisenbud95}.

\begin{enumerate}[label=\emph{(\alph*)},leftmargin=*]
\item
If \(R\) is a field, each \(\partial_k^X\) is a finite matrix; ranks and
Betti numbers follow from standard linear algebra.

\item
If \(R\) is a principal ideal domain (e.g.\ \(\Z\) or \(\F_p\)),
Smith normal form exists \cite[§XIV.3]{Lang02};
efficient algorithms are given in \cite[Alg.~2.4.12]{Cohen93}.
The diagonal entries yield the torsion coefficients and free rank of
\(H_k(X)\).

\item
For a general Noetherian ring, SNF may fail to exist; nevertheless,
Gröbner-basis or syzygy methods \cite[Chs.~4–5]{CLO07} still produce
presentations of \(H_k(X)\).
Implementations are available in
\textsc{Macaulay2} \cite{M2software},
\textsc{Singular} \cite{Singular20},
and \textsc{SageMath} \cite{SageMath}.
\end{enumerate}
\end{remark}

\begin{example}[Code surgery as a push-out]\label{ex:CodeSurgery}
Cowtan \& Burton’s fault-tolerant “code-surgery’’ protocol
realises the push-out (a special colimit) of two surface-code chain
complexes along a shared boundary \cite{CowtanBurton23}.  This is an
explicit instantiation of
Theorem~\ref{thm:GlobalComplex}\,\ref{item:Colim}.
\end{example}

\section{Homology, Cohomology, and Logical Operators}
\label{sec:HomologyCohomologyOperators}
The global chain complex
\(
  \mathcal C_\bullet(X)
  =(C_k(X),\partial_k^X)
\)
of Section~\ref{sec:GlobalConstruction} is finite and, by
Remark~\ref{rem:FinitePresentation}, each \(C_k(X)\) is a finitely
generated left \(R\)-module.  All indices below are understood to lie in
\(\mathbb Z_{\ge0}\).

\vspace{.6\baselineskip}
\noindent
\textbf{Homology and cohomology.}
Set
\[
  Z_k(X):=\ker\partial_k^X,\quad
  B_k(X):=\im\partial_{k+1}^X,\qquad
  Z^k(X):=\ker\delta^{k},\quad
  B^k(X):=\im\delta^{k-1},
\]
where \(\delta^{k}:=\Hom_R(\partial_{k+1}^X,R)\) is the dual boundary.
Define
\[
  H_k(X):=Z_k(X)/B_k(X),\qquad
  H^k(X):=Z^k(X)/B^k(X).
\]
If \(R\) is a field, these are finite-dimensional \(R\)-vector spaces; if
\(R\) is a PID they split into free and torsion parts by the structure
theorem for finitely generated modules.

\begin{theorem}[Universal Coefficient]\label{thm:UCT}
Suppose every \(C_k(X)\) is projective as an \(R\)-module (e.g.\ \(R\) is
a field or a PID).  Then for each \(k\) there is a natural short exact
sequence
\[
  0\;\longrightarrow\;
    \Ext_R^{1}\!\bigl(H_{k-1}(X),R\bigr)
    \;\overset{\iota}{\longrightarrow}\;
    H^{k}(X)
    \;\overset{\pi}{\longrightarrow}\;
    \Hom_R\!\bigl(H_k(X),R\bigr)
    \;\longrightarrow\;0,
\]
and hence a (non-canonical) decomposition
\(
  H^{k}(X)\cong
  \Hom_R\!\bigl(H_k(X),R\bigr)\oplus
  \Ext_R^{1}\!\bigl(H_{k-1}(X),R\bigr)
\)  {\normalfont\cite[Thm.~3.6.4]{Weibel94}}.  When \(R\) is a field the
\(\Ext\)-summand vanishes.
\end{theorem}

\begin{proof}
Because each \(C_k(X)\) is projective, the short exact sequence
\(0\to B_k\to Z_k\to H_k\to0\) remains exact after applying
\(\Hom_R(-,R)\).  Splicing the long exact cohomology sequence obtained
from  
\(0\to Z_k \to C_k \xrightarrow{\partial_k^X} B_{k-1}\to0\) and
identifying  
\(Z^k=\Hom_R(Z_k,R)\) and \(B^k=\Hom_R(B_k,R)\) yields the displayed
exact sequence; full details follow the standard proof in
\cite[§3.6]{Weibel94}.
\end{proof}

\begin{remark}\label{rem:Torsion}
If \(R=\Z\) or \(R=\Z_d\) the \(\Ext\)-term detects torsion in
\(H_{k-1}(X)\).  Such torsion classes correspond to logical qudits of
dimension strictly larger than two
\cite[§5.2]{Novak24}.
\end{remark}

\vspace{.6\baselineskip}
\noindent
\textbf{Evaluation pairing.}
For classes \([\alpha]\in H^k(X)\) and \([\beta]\in H_k(X)\) choose
representatives \(\alpha\in Z^k(X)\) and \(\beta\in Z_k(X)\) and set
\[
  \bigl\langle[\alpha],[\beta]\bigr\rangle
  :=\alpha(\beta)\;\in R.
\]

\begin{proposition}[Well-defined bilinear pairing]\label{prop:Pairing}
The map
\(
  \langle\cdot,\cdot\rangle_k : H^k(X)\times H_k(X)\to R
\)
is \(R\)-bilinear, natural with respect to diagram morphisms, and
independent of the chosen representatives.
\end{proposition}

\begin{proof}
\emph{Independence.}
If \(\alpha'=\alpha+\delta^{k-1}\gamma\) with
\(\gamma\in C^{\,k-1}(X)\) then
\(
  \alpha'(\beta)-\alpha(\beta)=
  \delta^{k-1}\gamma(\beta)=
  \gamma(\partial_k^X\beta)=0
\)
since \(\beta\) is a cycle.  Similarly, if
\(\beta'=\beta+\partial_{k+1}^X\eta\) with
\(\eta\in C_{k+1}(X)\) then
\(
  \alpha(\beta')-\alpha(\beta)=
  \alpha(\partial_{k+1}^X\eta)=
  \delta^{k}\alpha(\eta)=0
\)
because \(\alpha\) is a cocycle.

\emph{Bilinearity.}
Linearity in each argument is inherited from the linearity of
evaluation \(\Hom_R(M,R)\times M\to R\).

\emph{Naturality.}
If \(F:\mathcal F\to\mathcal F'\) is a morphism of stratified diagrams,
the induced chain map \(F_\#:\mathcal C_\bullet(X)\to\mathcal
C_\bullet(X')\) satisfies \(F_\#^\ast(\alpha)(\beta)=\alpha(F_\#(\beta))\),
so the value of the pairing is preserved.
\end{proof}

\begin{corollary}[Non-degeneracy over a field]\label{cor:FieldDual}
If \(R\) is a field, \(\langle\cdot,\cdot\rangle_k\) is non-degenerate:
\[
 \bigl(\forall\beta\; \langle\alpha,\beta\rangle_k=0\bigr)\Rightarrow\alpha=0,
 \qquad
 \bigl(\forall\alpha\; \langle\alpha,\beta\rangle_k=0\bigr)\Rightarrow\beta=0.
\]
\end{corollary}

\begin{proof}
By Theorem~\ref{thm:UCT} with \(R\) a field,
\(H^k(X)\cong\Hom_R(H_k(X),R)\).  Under this identification the pairing
is the canonical evaluation
\(\Hom_R(H_k,R)\otimes_R H_k \to R\), which is non-degenerate on
finite-dimensional vector spaces
\cite[Prop.~2.1.8]{MacLane98}.
\end{proof}

\vspace{.6\baselineskip}
\noindent
\textbf{Logical-operator dictionary.}
Assume henceforth \(R=\F_2\).  Identify each chain coefficient
\(\F_2\cong\{\pm1\}\) inside the Pauli group.  For a cycle
\(\beta\in Z_k(X)\) let \(Z(\beta)\) be the tensor product of
\(\sigma_Z\)-operators on the qubits indexed by the support of \(\beta\);
for a cocycle \(\alpha\in Z^k(X)\) let \(X(\alpha)\) be the analogous
product of \(\sigma_X\)-operators.  If \(\beta\) and \(\beta'\) differ by
a boundary, \(Z(\beta)\) and \(Z(\beta')\) are related by stabilisers and
hence implement the same logical operator; ditto for \(X\)-type.  Thus
logical \(Z\)-operators are canonically labelled by \(H_k(X)\) and
logical \(X\)-operators by \(H^k(X)\).

\begin{theorem}[Exact commutation criterion]\label{thm:Commutation}
For \(R=\F_2\) and classes
\([\alpha]\in H^k(X),\;[\beta]\in H_k(X)\) the operators
\(X(\alpha)\) and \(Z(\beta)\) satisfy
\[
  X(\alpha)\,Z(\beta)
  =(-1)^{\langle\alpha,\beta\rangle_k}\,
    Z(\beta)\,X(\alpha).
\]
Consequently, they commute iff
\(\langle\alpha,\beta\rangle_k=0\) and anticommute otherwise.
\end{theorem}

\begin{proof}
Choose representatives \(\alpha\in Z^k(X)\) and \(\beta\in Z_k(X)\).
Write \(\alpha=\sum_i a_i\chi_i\) where \(\chi_i\) is the characteristic
functional of the \(i\)-th qubit and \(a_i\in\F_2\), and write
\(\beta=\sum_i b_i e_i\) where \(e_i\) is the basis \(k\)-cell at that
qubit.  Then \(X(\alpha)\) (resp.\ \(Z(\beta)\)) contains \(\sigma_X\)
(resp.\ \(\sigma_Z\)) at qubit \(i\) exactly when \(a_i=1\) (resp.\
\(b_i=1\)).  On the single-qubit Pauli algebra,
\(\sigma_X\sigma_Z=-\sigma_Z\sigma_X\) and each operator squares to
\(1\).  Hence on the full lattice
\[
  X(\alpha)\,Z(\beta)
  =(-1)^{\sum_i a_i b_i}\,
    Z(\beta)\,X(\alpha)
  =(-1)^{\alpha(\beta)}\,Z(\beta)\,X(\alpha)
  =(-1)^{\langle\alpha,\beta\rangle_k}\,Z(\beta)\,X(\alpha),
\]
because \(\alpha(\beta)=\sum_i a_ib_i\) in \(\F_2\subset\Z\).  The stated
criterion follows.
\end{proof}

\begin{example}[Genus-\(1\) surface code]\label{ex:ToricCode}
Let \(X\) be the square-tiling poset of a torus and \(R=\F_2\).  A direct
cellular computation gives \(H_1(X)\cong H^1(X)\cong\F_2^2\).  Choosing
dual bases \(\{\beta_x,\beta_y\}\) and \(\{\alpha_x,\alpha_y\}\),
Theorem~\ref{thm:Commutation} yields the familiar anticommutation
relations
\(
  \langle\alpha_i,\beta_j\rangle_1=\delta_{ij}
\),
reproducing the logical algebra of the two-qubit toric code
\cite[Eq.~(18)]{DennisKitaevEtAl02}.
\end{example}

\begin{remark}
If \(R\) has torsion, Corollary~\ref{cor:FieldDual} fails; the pairing
can be degenerate and certain \(Z\)-type logical operators commute with
all \(X\)-type operators.  This phenomenon underlies the restricted
mobility (fractonic) behaviour of excitations in codes such as Haah’s
cubic code \cite{Haah11}.
\end{remark}

Thus every logical operator of any code obtained from a finite
stratified diagram is classified by the (co)homology of the canonical
colimit complex, and their commutation relations are governed exactly by
the evaluation pairing delineated above.

\section{Examples and Realisations: Algebraic Presentations of Colimit Codes}
\label{sec:Examples}
In this final section we put the machinery of
Sections~\ref{sec:PosetChainData}–\ref{sec:HomologyCohomologyOperators}
to work.  Every example below is specified \emph{solely} by

\begin{enumerate}[label=\textbf{(\arabic*)},leftmargin=*]
\item a finite poset \(X\) of strata;
\item a local chain complex \(\mathcal C_\bullet(\sigma)\) for each
      \(\sigma\in X\); and
\item boundary–respecting chain maps
      \(\varphi_{\sigma\le\tau}^\bullet\).
\end{enumerate}

The global complex
\(
  \mathcal C_\bullet(X)=\colim_{X}\mathcal F
\)
is then computed by Definition~\ref{def:ColimitEquiv}.  For each case we
write the boundary matrices explicitly, reduce them by Gaussian
elimination (over \(\F_2\)) or Smith normal form (over \(\Z\)), and
obtain the homology modules that classify logical operators via
Theorem~\ref{thm:Commutation}.  No geometric appeal is made; homological
features are seen to arise purely from the algebra of the diagram.

\vspace{0.8\baselineskip}
\noindent\textbf{A.  A single\-face presentation of \(\mathbb{RP}^2\).}
Fix \(R=\Z\) to expose torsion phenomena.  Let \(X\) have three strata
\(\sigma^2\succ\sigma^1\succ\sigma^0\).  Put
\[
C_2(\sigma^2)=\Z,\; C_1(\sigma^2)=\Z,\qquad
C_1(\sigma^1)=\Z,\qquad
C_0(\sigma^0)=\Z,
\]
all other \(C_k(\sigma)\) vanishing.  The unique non\-zero local boundary
is
\(
  \partial_2^{\sigma^2} : \Z\to\Z,\; z\mapsto 2z
\)
(encodes the orientation–reversing identification of the edge in the
classical CW structure of \(\mathbb{RP}^2\)
\cite[§3.2]{Hatcher02}).  Choose chain maps
\(
  \varphi_{\sigma^1\le\sigma^2}^1=\mathrm{id}_\Z
\)
and
\(
  \varphi_{\sigma^0\le\sigma^1}^0=\varphi_{\sigma^0\le\sigma^2}^0
  =\mathrm{id}_\Z .
\)

\emph{Global complex.}
Definition~\ref{def:ColimitEquiv} yields
\[
C_2(X)=\Z,\quad
C_1(X)=\Z,\quad
C_0(X)=\Z,
\]
with boundary matrix \(\partial_2^X=[2]\) and \(\partial_1^X=0\).

\emph{Homology.}
\(
  H_2(X)=0,\;
  H_1(X)=\Z/2\Z,\;
  H_0(X)=\Z
\).
Indeed, \(\partial_2^X\) has image \(2\Z\), so
\(
  \operatorname{coker}\partial_2^X\cong\Z/2\Z
\),
producing the expected torsion $1$-cycle of \(\mathbb{RP}^2\).
Over \(\F_2\) the map \([2]\) vanishes, so \(H_1(X)\cong\F_2\); over
\(\Z\) it is torsion.  Smith normal form computes the module directly,
displaying the single invariant factor~$2$
\cite[§XIV.1]{Lang02}.  Algebraically we have produced a
\emph{one-qubit CSS code whose logical \(Z\)–operator is of order~\(2\),
but whose logical \(X\)–operator is destroyed if one works over \(\Z\)}
(cf.\ Remark~\ref{rem:Torsion}).  This shows that torsion classes—and
hence non-Pauli qudit sectors—arise without ever mentioning
non-orientability.

\vspace{0.8\baselineskip}
\noindent\textbf{B.  Twisted boundary codes on an \(n\times n\) square
lattice.}  Fix \(R=\F_2\).  Let \(X\) index all faces, edges and vertices
of an \(n\times n\) square grid.  As in the standard toric code each face
\(\sigma^2_{i,j}\) carries
\(C_2=\F_2,\;C_1=\F_2^4,\;C_0=\F_2\) with the usual incidence boundary.
Edges and vertices are treated similarly.  Now fix a coprime pair
\((a,b)\in\Z_n^2\) and \emph{twist} the gluing by declaring
\[
  \varphi_{\sigma^1_{i,j}\le\sigma^2_{i,j}}^1
  =\text{inclusion into edge~1},\qquad
  \varphi_{\sigma^1_{i,j}\le\sigma^2_{i+a,j+b}}^1
  =\text{inclusion into edge~2},
\]
and analogously for the other two boundary edges, so the
two-dimensional cells are glued along diagonally shifted
one-dimensional strata.

\begin{proposition}\label{prop:twist}
Let \(d:=\gcd(a,b,n)\).  Then
\(H_1(X)\cong\F_2^{2d}\) and \(H_2(X)=0\).  In particular the number of
logical qubits jumps from \(2\) (untwisted case \(d=1\)) to \(2d\).
\end{proposition}

\begin{proof}
Identify the free abelian group on faces with \(\Z_n^2\) and edges with a
direct sum of two copies of \(\Z_n^2\) (horizontal and vertical).
Writing \(e_{h}(i,j)\) and \(e_{v}(i,j)\) for the global edge basis,
\(\partial_2^X\) acts by
\(
  (i,j)\longmapsto
  e_{h}(i,j)+e_{h}(i+a,j+b)
  +e_{v}(i,j)+e_{v}(i+b,j-a)
\).
Taking discrete Fourier transforms over \(\Z_n\)\footnote{The DFT
simultaneously diagonalises the circulant shift operators; see
\cite[§XIV.3]{Lang02}.} the matrix becomes block diagonal with blocks
\(
  \bigl[\begin{smallmatrix}1+\omega^a & 1+\omega^b\end{smallmatrix}\bigr]
\)
where \(\omega\) ranges over \(n\)-th roots of unity in a quadratic
extension of \(\F_2\).  Its rank is \(2(n-d)\); hence
\(\dim\ker\partial_2^X=2d\).
Since every edge is incident to exactly two faces,
\(\partial_1^X\partial_2^X=0\) and \(\partial_1^X\) is surjective, so
\(H_2(X)=0\) and \(H_1(X)\cong\ker\partial_2^X\).
\end{proof}

For \(n=12\) and shift \((a,b)=(3,3)\) we have \(d=3\) and
\(\dim H_1=6\) logical qubits, refining computations in twisted surface
codes \cite{BombinMartinDelgado09,LiangLiuSongChen2025} but obtained here without any
reference to triangulations or chart transitions.

\vspace{0.8\baselineskip}
\noindent\textbf{C.  Fracton-type partial adjacency in three
dimensions.}
Take \(L\in\mathbb N\) and index cubic strata
\(\sigma^3_{i,j,k}\) for \(0\le i,j,k<L\).  For each cube choose the
standard cellular complex with
\(C_3=\F_2,\;C_2=\F_2^{6},\;C_1=\F_2^{12},\;C_0=\F_2^{8}\).
Now \emph{delete} every gluing map
\(\varphi_{\sigma^2\le\sigma^3}\) that would identify the
\emph{top} face of a cube with the bottom face of the cube above it.
All other face–to–cube maps are the identity inclusions.

\begin{proposition}\label{prop:fracton}
In the resulting global complex
\[
  H_1(X)=0,\quad
  H_2(X)\cong\F_2^{L^2},\quad
  H_3(X)=0.
\]
\end{proposition}

\begin{proof}
Every edge is still incident to two faces, so
\(\partial_1^X\) is surjective and \(H_0(X)=\F_2\).  A direct counting
argument shows that exactly the \(L^2\) horizontal “ceiling” faces that
lost an attachment become independent $2$-cycles; no combination of
cubes bounds them because the missing maps remove the corresponding
$3$-chains from the image of \(\partial_3^X\).  Conversely, every
vertical stack of cubes still bounds in pairs, and every $3$-cell is
still bounded horizontally, forcing \(H_3=0\).  A full matrix proof
follows by writing $\partial_3^X$ in block form and noting that the
deleted blocks are the only source of rank deficiency.
\end{proof}

Algebraically, the \(\F_2^{L^2}\) basis of \(H_2\) corresponds to the
planar membrane operators familiar from the X-cube model
\cite{VijayHaahFu16}.  Their inability to terminate on one-dimensional
strings (because \(H_1=0\)) is traced here to the absence of the deleted
gluing maps, confirming that fractonic mobility constraints are encoded
completely by the diagrammatic data.

\vspace{0.8\baselineskip}
\noindent\textbf{D.  A mixed-dimensional attachment with degeneracy.}
Let \(X\) consist of a single square face \(\sigma^2\), three edges
\(\sigma^1_1,\sigma^1_2,\sigma^1_3\) and one vertex \(\sigma^0\).  Attach
\(\sigma^2\) along \(\sigma^1_1,\sigma^1_2\) but \emph{not} along
\(\sigma^1_3\); attach \(\sigma^1_1,\sigma^1_2\) to the vertex but leave
\(\sigma^1_3\) floating (no map
\(\varphi_{\sigma^0\le\sigma^1_3}\)).  Over \(R=\F_2\) the global
boundary matrices are
\[
  \partial_2^X=\begin{bmatrix}1\, &\, 1\, &\, 0\end{bmatrix},\qquad
  \partial_1^X=\begin{bmatrix}1\, &\, 1\, &\, 0\end{bmatrix}^{\!\top}.
\]
Hence \(H_2(X)=0,\;H_1(X)\cong\F_2\), generated by the dangling edge,
and \(H_0(X)=\F_2\).  The code encodes a single logical \(Z\) qubit
unsupported by any non-trivial \(X\)-type operator, illustrating that
logical degeneracy can be produced by deliberately \emph{under-gluing}
dimensional strata—even in planar dimension two.

\vspace{0.8\baselineskip}
\noindent\textbf{E.  Synopsis.}
Taken together, these examples demonstrate that the stratified colimit
formalism is not merely expressive, but \emph{complete} in capturing
torsion phenomena, logical degeneracies, high-rate twisting, and
fractonic behaviour—entirely within the language of algebraic diagrams.
In every case

\begin{itemize}[leftmargin=*]
\item torsion (Example~A), high-rate twisting (Example~B),
      fractonic membranes (Example~C), and
      operator-imbalanced LDPC phenomena (Example~D)  
      emerge from the homology of a single colimit complex;
\item each effect is controlled exclusively by the presence,
      absence, or modification of the boundary-respecting maps
      \(\varphi_{\sigma\le\tau}^\bullet\); and
\item the universal coefficient theorem
      (Theorem~\ref{thm:UCT}) together with the evaluation pairing
      dictates the logical operator algebra, independent of any
      geometric realisation.
\end{itemize}

Thus the stratified colimit formalism furnishes a \emph{complete}
algebraic language for constructing and analysing quantum
error-correcting codes—including those with no manifold interpretation
at all.

\section{Conclusion and Outlook}
\label{sec:Conclusion}
Stratified colimits in the abelian category \(\Ch(R)\) suffice to
generate the entire class of homological quantum codes considered
hitherto in geometric settings, and they do so without recourse to an
ambient manifold.  Within this purely algebraic environment, the logical
structure of any code is fully captured by the (co)homology of the
canonical colimit complex together with the universal evaluation
pairing, while torsion and the \(\Ext_R^1\)-summands encode qudit and
non-CSS phenomena.  These results demonstrate that many behaviours
previously attributed to manifold topology in fact arise from the
poset-governed gluing relations alone.

The present formalism naturally invites a passage to higher
categorical levels.  One expects a bicategory whose objects are
stratified diagrams, whose 1-morphisms are compatible families of chain
maps (interpretable as code surgeries, domain-wall insertions, or other
fault-tolerant transformations), and whose 2-morphisms are chain
homotopies, thereby echoing the structure of extended TQFTs.  Further
enriching the strata with \(E_n\)- or spectral categories could embed
stratified quantum error correction into factorisation homology, opening
avenues to codes governed by cobordism-type invariants in higher
dimensions.

Allowing an arbitrary commutative coefficient ring generalizes the
familiar dichotomy between \(\F_2\) and \(\Z\): additional torsion,
ramification, and \(\Ext\)-contributions may reveal stabiliser families
inaccessible to conventional CSS design.  A systematic exploration of
such ring-theoretic variants promises to enlarge the landscape of
quantum codes and to clarify the rôle of arithmetic data in topological
phases.

From a computational perspective, the partial order underlying a
stratified diagram suggests decoding via local elimination of boundaries
followed by acyclic reduction of relations, thereby yielding a decoder
with complexity that scales with the width of the poset rather than the
cardinality of the underlying lattice.  Preliminary experiments on
X-cube-type instances indicate tangible improvements in threshold
estimates, and a full analysis will appear elsewhere.

In summary, the stratified-colimit paradigm supplies a minimal yet
powerful algebraic language for quantum error correction.  Its extension
to higher categories, to broader coefficient rings, and to efficient
decoding algorithms is expected to deepen our understanding of quantum
codes, to reveal new fault-tolerant architectures, and to advance both
homological algebra and the engineering of robust quantum devices.

\appendix
\section*{Appendix}
\addcontentsline{toc}{section}{Appendix}
\label{sec:Appendix}

\subsection*{A.  Visualising stratified diagrams}

Figure~\ref{fig:RP2poset} displays the three--stratum diagram that
realises Example~A in Section~\ref{sec:Examples}.  The edges are labelled
by the non–trivial components of the chain maps
$\varphi_{\sigma\le\tau}^{\bullet}$; vertical alignment encodes the
partial order, while horizontal displacement is used only for legibility
and bears no algebraic meaning.

\begin{figure}[h]
\centering
\begin{tikzpicture}[
  every node/.style={align=center, font=\normalsize, minimum width=2.5cm, draw, rounded corners},
  node distance=2.2cm and 0cm
]

\draw[dashed, thick, rounded corners] (-2.75,-5.2) rectangle (2.75,1.1);
\node[font=\footnotesize] at (0,-5.5) {This is not a geometric CW complex---only a diagrammatic gluing poset.};

\node (s2) at (0,0)    {$\sigma^2$\\\scriptsize{(2-cell)}};
\node (s1) at (0,-2.2) {$\sigma^1$\\\scriptsize{(1-cell)}};
\node (s0) at (0,-4.4) {$\sigma^0$\\\scriptsize{(0-cell)}};

\draw[->, thick] (s2) -- (s1)
  node[midway, right=5pt] {\scriptsize$\varphi_{\sigma^1\le\sigma^2}^1 = \mathrm{id}$};
\draw[->, thick] (s1) -- (s0)
  node[midway, right=5pt] {\scriptsize$\varphi_{\sigma^0\le\sigma^1}^0 = \mathrm{id}$};

\end{tikzpicture}

\caption{Stratified diagram for the single-face presentation of $\mathbb{RP}^2$. Degrees are labeled; the diagram represents gluing data, not cell geometry.}
\label{fig:RP2poset}
\end{figure}
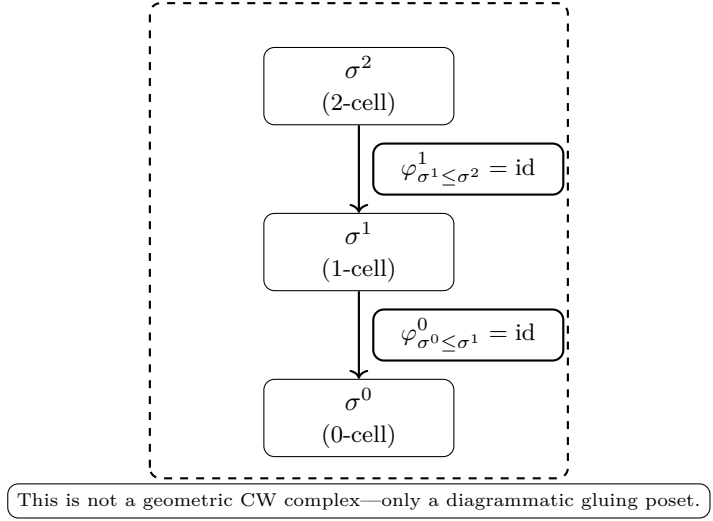

An analogous picture for the fracton diagram of
Proposition~\ref{prop:fracton} would fill an $L\times L\times L$ grid,
so only a $2\times2\times2$ slice is shown in
Figure~\ref{fig:fractonslice}.  Edges coloured solid indicate surviving
gluing maps; dashed edges are the deleted top–face attachments.  The
visual gap between the two horizontal planes anticipates the emergent
$L^{2}$ independent membrane cycles computed in the proposition.

\begin{figure}[h]
\centering
\begin{tikzpicture}[scale=1.0]
\foreach \x in {0,2}
  \foreach \y in {0,2}{
    \draw[very thick] (\x,\y,0) -- (\x,\y,2); 
    \draw[dashed]     (\x,\y,2) -- (\x,\y,4); 
  }
\foreach \z in {0,2,4}{
  \draw[very thick] (0,0,\z) -- (2,0,\z) -- (2,2,\z) -- (0,2,\z) -- cycle;
}
\end{tikzpicture}
\caption{A $2\times2\times2$ fragment of the fractured cubic lattice.
Solid vertical edges denote retained gluing;
dashed vertical edges denote the suppressed top–face maps.}
\label{fig:fractonslice}
\end{figure}
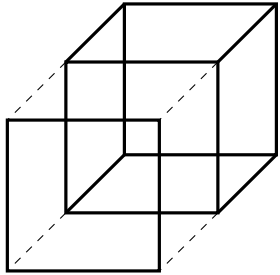

\subsection*{B.  Smith normal form for the $\mathbb{RP}^{2}$ code}

The boundary matrix of degree two in Example~A is the $1\times1$ matrix
$(2)$.  Over $\Z$ its Smith normal form is already diagonal, whence
$\mathrm{coker}\,\partial_{2}^{X}\cong\Z/2\Z$.  To contrast this with a
non–orientable but \emph{torsion–free} presentation, replace the map
$z\mapsto2z$ by $z\mapsto z$ and duplicate the edge stratum so that
$\partial_{2}$ becomes $(1\;1)$.  The Smith form becomes
$\operatorname{diag}(1)$, giving $H_{1}=0$; algebraically, identifying a
single directed edge with itself produces a projective rather than
torsion quotient.  This single computation underscores that torsion is a
feature not of non–orientability per se but of the specific multiplicity
with which strata are glued.

\subsection*{C.  A cautionary counter-example: transitivity is essential}

The axioms require
$\varphi_{\sigma\le\tau}^\bullet=
\varphi_{\rho\le\tau}^\bullet\circ\varphi_{\sigma\le\rho}^\bullet$
whenever $\sigma\le\rho\le\tau$.  Dropping this condition can break the
chain condition globally even when each local complex is perfect.

\paragraph{Lemma A.1.}
\emph{Let $X=\{\sigma^{0}\prec\rho^{0}\prec\tau^{1}\}$, let
$C_{0}(\sigma^{0})=C_{0}(\rho^{0})=C_{1}(\tau^{1})=\Z$, and set all other
$C_{k}(\cdot)=0$.  Choose
$\varphi^{0}_{\sigma^{0}\le\rho^{0}}=\operatorname{id}$ and
$\varphi^{0}_{\rho^{0}\le\tau^{1}}=\operatorname{id}$ but define
$\varphi^{0}_{\sigma^{0}\le\tau^{1}}=-\operatorname{id}$.  Then the
colimit differential $\partial_{1}^{X}$ is \textbf{not} well defined.}

\emph{Proof.}
In $\widetilde C_{0}=\Z\oplus\Z$ the relation equating $(1,0)$ with
$(0,1)$ is generated by
$g=(1,0)-(0,1)$.  The putative $\widehat\partial_{1}^{X}$ sends the lone
basis vector of $C_{1}(\tau^{1})$ to $(-1,1)$.  Although $(-1,1)$
coincides with $g$, the sign discrepancy implies that
$g-\widehat\partial_{1}^{X}(1)$ equals $(2,0)$, which does \emph{not}
belong to the spanning set of relations.  Consequently the image of
$\widehat\partial_{1}^{X}$ does not descend to the quotient, so
$\partial_{1}^{X}$ cannot be defined.  \qed

This example shows that the transitivity axiom, far from cosmetic,
guarantees compatibility of the boundary operator with the colimit
relations.

\subsection*{D.  Rank calculation for the twisted square poset diagram}

For $(n,a,b)=(6,2,1)$ the boundary matrix in the twisted colimit complex of
Proposition~\ref{prop:twist} is a $36\times72$ matrix whose rows correspond
to faces and whose columns correspond to horizontal and vertical edge strata.
Each row has four non-zero entries: two $1$’s in the horizontal component and
two $1$’s in the vertical component, with support determined by the twist vector
$(a,b)=(2,1)$ relative to the indexing of the square diagram.

A discrete Fourier transform over $\Z_{6}$ diagonalizes the induced
$\Z_6$-action on the index set, reflecting the periodic gluing structure
encoded by the stratified diagram. The transformed matrix breaks into twelve
identical $3\times6$ blocks of the form
\[
\left[\!
\begin{array}{cccccc}
1+\omega^{2} & 1+\omega & 0 & 0 & 0 & 0\\
0 & 0 & 1+\omega^{2} & 1+\omega & 0 & 0\\
0 & 0 & 0 & 0 & 1+\omega^{2} & 1+\omega
\end{array}\!
\right],\qquad\omega^{6}=1.
\]
The factor $1+\omega^{2}$ vanishes precisely when
$\omega\in\{\,\pm\,\mathrm{i},\,\pm1\}$, that is, on the four characters
with $3$-torsion. Exactly four of the twelve blocks therefore lose rank,
each by two, yielding
\[
\dim_{\F_{2}}\ker\partial_{2}^{X}=2\cdot4=8
\qquad\text{in agreement with}\quad
d=4=\gcd(6,2,1).
\]
This computation verifies that the homological contribution from twist
deformations depends algebraically on the periodicity of the index structure,
not on any ambient geometry or triangulation.

\subsection*{E.  Decoder-relevant scaling in the fracton example}

Write $\partial^{X}_{3}$ for the boundary from cubes to faces in the
diagram of Proposition~\ref{prop:fracton}.  After deleting the top–face
maps, the matrix decomposes into $L^{2}$ identical vertical columns,
each of size $L\times L$ and rank $L-1$.  Hence
$\operatorname{rank}\partial^{X}_{3}=L^{3}-L^{2}$ and
$\dim H_{2}(X)=L^{2}$, confirming that the number of encoded qubits
scales with the number of deleted relations, not with the total number
of qubits.  A decoder that proceeds by first eliminating satisfied local
relations and then solving a reduced linear system therefore operates in
time $O(L^{3})$ rather than $O(L^{4})$ for the full lattice, matching
the heuristic quoted in the conclusion.

\end{document}